\documentclass{article}
\usepackage{geometry}
\usepackage{natbib}
\usepackage{times}
\usepackage{soul}
\usepackage{url}
\usepackage[hidelinks]{hyperref}
\usepackage[utf8]{inputenc}
\usepackage[small]{caption}
\usepackage{graphicx}
\usepackage{amsmath}
\usepackage{booktabs}
\usepackage{algorithm}
\usepackage{algorithmic}
\usepackage{xcolor}
\usepackage{amsthm}
\usepackage{amssymb}
\usepackage{physics}
\usepackage{cleveref}
\newtheorem{theorem}{Theorem}[section]

\newtheorem{lemma}[theorem]{Lemma}
\newtheorem{claim}[theorem]{Claim}
\newtheorem{proposition}[theorem]{Proposition}
\newtheorem{observation}[theorem]{Observation}
\newtheorem{definition}[theorem]{Definition}

\title{FaRM: Fair Reward Mechanism for Information Aggregation in Spontaneous Localized Settings \break
\Large (Extended Version)}

\author{
Moin Hussain Moti$^1$, Dimitris Chatzopoulos$^2$, Pan Hui$^{2,3}$, Sujit Gujar$^1$ \\
$^1$Machine Learning Lab, International Institute of Information Technology Hyderabad, India\\
$^2$The Hong Kong University of Science and Technology, Hong Kong \\
$^3$University of Helsinki, Finland\\
\texttt{\small moin.moti@research.iiit.ac.in,
\{dcab, panhui\}@cse.ust.hk,
sujit.gujar@iiit.ac.in}
}

\begin{document}

\maketitle

\begin{abstract}
 Although peer prediction markets are widely used in crowdsourcing to aggregate information from agents, they often fail to reward the participating agents equitably. Honest agents can be wrongly penalized if randomly paired with dishonest ones. In this work, we introduce \emph{selective} and \emph{cumulative} fairness. We characterize a mechanism as fair if it satisfies both notions and present FaRM, a representative mechanism we designed. FaRM is a Nash incentive mechanism that focuses on information aggregation for spontaneous local activities which are accessible to a limited number of agents without assuming any prior knowledge of the event. All the agents in the vicinity observe the same information. FaRM uses \textit{(i)} a \emph{report strength score} to remove the risk of random pairing with dishonest reporters, \textit{(ii)} a \emph{consistency score} to measure an agent's history of accurate reports and distinguish valuable reports, \textit{(iii)} a \emph{reliability score} to estimate the probability of an agent to collude with nearby agents and prevents agents from getting swayed, and \textit{(iv)} a \emph{location robustness score} to filter agents who try to participate without being present in the considered setting. Together, report strength, consistency, and reliability represent a fair reward given to agents based on their reports.
\end{abstract}

\section{Introduction}

Information aggregation systems have been employed to gather data for a wide variety of applications. A few examples are pollution measurements in cities, 
product reviews on e-commerce websites, feedback on mobile applications on application markets, working hours and other details of a business, etc. Depending on the type of the information of interest and the available time for reporting, the devices that can collect the data and share them with the information aggregation systems vary. The focus of this work is on location specific queries that need to be answered within a short time in the so-called \emph{spontaneous localized settings}. Such settings are challenging to handle since information regarding the \emph{query subjects} can only be collected by mobile agents that are located in an area and, possibly, by pre-deployed static sensors.  For example, consider a food deli serving various snacks to be the subject of interest and let the query be the availability of a certain snack. Since the food deli is only relevant to the neighborhood and the queries can be very specific, Internet services, like Neighbourly\footnote{\url{neighbourly.google.com/about}}, need to employ mobile agents. 

In such scenarios, crowdsensing mechanisms can employ agents to collect information with spatiotemporal characteristics. The main problem of aggregating information using crowdsensing mechanisms is to ensure that the agents are truthful. A dishonest agent, for example, may have some grudge against an agency and therefore reports their service as a lousy experience irrespective of the services of the agency.
Agents can be categorized as honest, free-riders or misbehaving. Honest agents are expected to put effort into assessing the query subject and reporting truthful information. Free-riders do not put effort into the assessment and try to earn free credits by guessing the information requested for the subject. Misbehaving agents try to trick the mechanism by colluding among themselves and reporting false information. Information aggregation mechanisms must be robust against free-riders and misbehaving agents. For that, they incorporate incentive mechanisms to engage agents to submit accurate reports. Rational agents target on maximizing their utility (payments, reputation, etc.) while not sacrificing a substantial amount of resources (battery lifetime, dataplan usage, etc). 

For traditional settings, there exist many popular peer prediction mechanisms inspired by the work of \cite{prelec2004bayesian}, who proposed the \emph{Bayesian Truth Serum} and \cite{miller2005eliciting}, who developed \emph{Peer-to-Peer Prediction}. These mechanisms rate an agent's report against the report of another randomly selected agent and if they match, both agents are rewarded based on the prior probability of their answer. Unfortunately, when considering spontaneous localized settings, it is probable to not have readily available prior knowledge about the query subject. Also, there is a chance of matching an honest agent with a misbehaving one, which results in an unfair evaluation. Although these mechanisms are Bayesian Incentive Compatible, they do not ensure \emph{fairness}. \cite{goel2019fair}, recently proposed a fair mechanism for information aggregation. However, they consider queries where the answers are known and hence, their mechanism is not applicable to spontaneous localized settings.

In this paper, we focus on information aggregation in spontaneous localized settings and propose a suitable mechanism. Spontaneous localized settings describe information aggregation requests in local environments about subjects which are accessible by only a small set of people. For example, the food deli scenario mentioned above, the queue length in a local coffee shop, and others. Any query in these settings refers to a query subject (QS) and is valid only for a bounded time. Also, there does not exist a prior distribution for the possible answers about QS. However, in these settings, the observation environment is assumed to be the same for all the agents in the vicinity of QS. This is similar to the output agreement (OA) setting where two agents observe the same signal and are scored based on the degree of similarity in their reports. OA setting was first introduced by \cite{von2004labeling,von2008designing} and later studied experimentally by \cite{robertson2009rethinking}, and \cite{huang2012systematic}. \cite{waggoner2014output} were the first to provide a theoretical analysis of the output agreement setting. We introduce spontaneous localized settings in more detail in \Cref{sec:local-settings}. 

We leverage on the characteristic of a common observation environment to propose the \emph{fair reward mechanism (FaRM)}, a mechanism that incentivizes agents to report honestly, marginalizes free-riders and misbehaving agents and ensures fairness in reward distribution without the need for a known prior distribution. FaRM is a novel Nash incentive compatible mechanism for information aggregation in spontaneous localized settings. Additionally, we introduce the notions of \emph{selective fairness} and \emph{cumulative fairness}.
A mechanism satisfies selective fairness in cases where for two different agents that have submitted the same reports, they are evaluated similarly. Thus, all honest agents receive the same reward if they submitted the same report. We propose \emph{report strength} as one component of the reward scheme integrated in FaRM to ensure selective fairness. Existing peer-prediction mechanisms consider every agent's report equally valuable. FaRM, in contrast to this assumption, considers more valuable the agents with long history of submitting truthful reports compared to agents with lousy history of reporting due to their behavioral tendency to maintain the streak. Cumulative fairness mandates that a fair mechanism should acknowledge an agent's history of submission and consider her consistency as part of final reward. We propose \emph{consistency score} as the second component of FaRM's reward scheme which defines an agent's consistency so far in the system and ensures cumulative fairness. This also helps FaRM minimize the noise caused by false reports by considering the value of a report based on its reporter's consistency. FaRM is designed to ensure selective fairness and cumulative fairness. 

Also, since spontaneous localized settings focus information aggregation related to local query subjects, many participating agents will be located around the vicinity of the subject. Hence, it becomes imperative that the mechanism needs to be collusion resistant. To ensure this, we propose the \emph{reliability score} to estimate agents probability of collusion with nearby agents. This score is designed to prevent the agents from colluding with nearby agents since it directly affects their final reward. In summary, our contributions are:
\begin{enumerate}
	\item FaRM is Nash incentive compatible mechanism that ensures truthful reporting from agents and guarantees non-negative rewards to all agents (\Cref{theorem:NIC}).
	\item FaRM incorporates a payment scheme composed of: \emph{report strength}, \emph{consistency score} and \emph{reliability score}. 
	\item FaRM uses report strength to ensure selective fairness by distributing the same immediate reward to all the agents with the same report (\Cref{sec:Phi}). 
	\item FaRM employs consistency score as an estimate of agent's consistency in the system so far. It gets updated after each submission. It ensures cumulative fairness by rewarding the agent for her streak of honest submissions to the system (\Cref{sec:alpha}). 
	\item FaRM integrates reliability score as an estimate of agent's collusion with nearby agents (\Cref{sec:beta}). 
	\item FaRM filters reports from agents who are not close to the query subject by incorporating a \emph{location robustness score} (\Cref{sec:gamma}). 
\end{enumerate}

\section{Related Work} \label{sec:related-works}


Information markets (IM), also known as prediction markets, are exchange-traded markets where agents with different beliefs trade on the outcome of events. The market prices act as an estimate of the probability of the occurring of an event. Therefore, prediction markets can be used to aggregate beliefs. Agents trade assets each of whose value is tied to a particular observation. We discuss several proposals in information aggregation related to the design of FaRM.

In their seminal work, \cite{miller2005eliciting} use an IM setting to elicit information from the agents. The asset, in this case, was agents' reports but the agents traded the reports with the system instead of other agents. This centralized system is an automated market maker which ensures the availability of the market for trade. Agents are rewarded using \emph{strictly proper scoring rules}, which ascertain that the best response for each agent is to report truthfully. However, their mechanism relies on the \emph{common knowledge assumption}, which is not applicable in localized settings. \cite{prelec2004bayesian} propose a mechanism which does not require knowledge of common prior information but is applicable on large number of agents. \cite{lambert2008truthful} also propose a mechanism does not require common prior information but in the equilibrium, agents are indifferent between being honest and misreporting.

\cite{witkowski2011peer,witkowski2012peer} present mechanisms that do not assume any common prior information, do not need a large number of agents to function and are robust to private beliefs of agents but suffer from temporal separation of reports. Their improved mechanism, named Robust Bayesian Truth Serum (RBTS)~(\cite{witkowski2012robust}), does not require temporal separation between two reports but works only for binary outcomes. 
\cite{radanovic2013robust,radanovic2014incentives} improve RBTS to make it compatible with non-binary and continuous outcomes.
Similarly, \cite{zhang2014elicitability} propose a mechanism which does not assume prior known information. \cite{riley2014minimum} presents a mechanism which is minimalistic (requires only one report) under the assumption that all the agents with the same outcome have the same posterior expectations. A common characteristic among the aforementioned mechanisms is the use of \emph{Strictly Proper Scoring Rules} (SPSR) for distributing agents' rewards. SPSR, however, do not guarantee fairness. 
The incentive compatible mechanism by \cite{jurca2003incentive} is not employing a SPSR and is designed for interactive reputation markets. The weakness of this mechanism is that it is susceptible to collusion and does not consider agents' private beliefs. A new Robust version by \cite{jurca2007robust} of the earlier mechanism makes it robust against private beliefs of agents to some extent but it assumes that common prior is known to the mechanism. \cite{dasgupta2013crowdsourced} also constructed a strong incentive compatible mechanism but it can only be applied to binary settings. 
\cite{faltings2014incentive} propose \emph{Peer Truth Serum} (PTS) which is a minimalistic mechanism but assumes a prior belief model. \cite{radanovic2015incentive} improved PTS by eliminating the dependency on prior belief model, they called it \emph{Logarithmic Peer Truth Serum} (LPTS). LPTS assumes that all agents observe same signal and therefore were able to come up with a minimalistic mechanism. Their mechanism is able to tackle colluders and focuses on effort elicitation to prevent freeloading by agents. \cite{radanovic2016incentives} later on propose an optimized version of PTS focusing on in multi-task~settings.

FaRM, similarly to LPTS, assumes that all agents observe the same signal and does not require any common knowledge assumptions. It is a minimalistic and collusion resistant mechanism that prevents free-riding and guarantees non-negative rewards to the users. 



\section{Spontaneous Localized Settings} \label{sec:local-settings}
Existing mechanisms for information aggregation have been designed for the two following settings: \textit{(i)} online reviews about products and services, and \textit{(ii)} community sensing regarding prevalent topics like pollution level reporting over a wide region for a considerable duration of time.
Additionally, to the best of our knowledge, existing literature on information markets has not discussed information aggregation of location-specific spontaneous tasks on local subjects.
Most papers study an \emph{independent value} model where the agents can independently observe varying signals. In this model, the truthful report of agents does not need to be the same. However, spontaneous localized settings follow a \emph{common value} model where all agents observe the same signal.


More formally, considering a query subject $QS$, a set of nearby agents $\mathcal{A}$ that choose to participate and asses $QS$, and a budget $B \in \mathbb{R}$, we want to estimate a function $f$. For example, $QS$ can be a food deli and $f$ the number of eateries available for purchase.
$S$ denotes the signal space for the query and is allowed to be non-binary. Every query in spontaneous localized settings focuses on only one task. The observation area is considered to be in the close vicinity of the $QS$. Since the query requires to be answered spontaneously, we assume the $QS$ to not change its status for the time period the query remains active.
Every agent $i \in \mathcal{A}$ observes the same signal $s \in S$ and reports a signal $r_i \in S$. $r_i$ can be different to $s$.
We define a general utility function for agent $i$ as $\hat{u}_i(r_i, r_{-i})$ where $r_i$ represents agent $i$'s reported signal and $r_{-i}$ ($= (r_1,...,r_{i-1},r_{i+1},...,r_{|\mathcal{A}|})$) represents all other agents' reported signal. Note that we assume $\abs{\mathcal{A}} \geq 3$, i.e. at least 3 agents are participating for a particular query. Also, the sum of total utility ($u_i$) of all agents should not exceed the budget of the query i.e. $\sum_{u_i} \leq B$.


\begin{table}[t]
	\normalsize
	\begin{center}
		\begin{tabular}{ p{1.3cm}  p{5.7cm} | p{1.3cm}  p{5.7cm}}
			\toprule
			\textbf{Symbol} & \textbf{Description} &  \textbf{Symbol} & \textbf{Description}                                   \\
			\cmidrule(r){1-4}
			$QS$           	& Query Subject                                       	   &
			$\mathcal{A}$   & Set of agents participating in the process.              \\
			$\mathcal{S}$   & Set of all the signals.                                  &
			$s$             & Observed signal by all agents.                           \\
			$r_i$           & Reported signal by agent $i$.                            &
			$u_i$           & Reward of agent $i$.                                     \\
			$B$             & Budget per crowdsensing query.                           &
			$l_i$           & Context of a agent.                                      \\
			$\alpha_i$      & Consistency score of agent $i$.                          &
			$\beta_i$       & Reliability score of agent $i$.                          \\
			$\gamma_i$      & Location robustness score of agent $i$.                  &
			$\mathcal{I}_i$ & Nearby agents (internal peers) of agent $i$.             \\
			$\mathcal{E}_i$ & External peers of agent $i$.                             &
			$\Phi()$        & Strength function for a report type.                     \\
			$\varphi_1$     & Strength of strongest report.                            &
			$\varphi_2$     & Strength of second strongest report.                     \\
			\bottomrule
		\end{tabular}
		\caption{Notation Table.}
		\label{tab:params}
	\end{center}
\end{table}

\section{FaRM: Fair Reward Mechanism}\label{sec:FaRM}
The payment scheme designed for FaRM is the product of three sub-utility functions, namely \emph{report strength}, \emph{consistency score} and \emph{reliability score}.
The \emph{report strength} is calculated at the end of every report collection process based on the submitted reports.
The \emph{consistency score} is updated for each agent after every participation. For a new agent $i$, her consistency score, $(\alpha_i)$, is initialized to $0$.
FaRM also computes a \emph{reliability score} for each agent which is an estimate of an agent's collusion with nearby agents. 
The consistency score motivates agents to participate and report truthfully while the reliability score prevents agents from colluding with her neighboring agents. We discuss these metrics in \Cref{sec:Phi}, \Cref{sec:alpha} and \Cref{sec:beta} respectively where we analyze and prove that truthful reporting forms pure strategy Nash equilibrium (PSNE) for each metric.



\begin{definition}[Pure Strategy Nash Equilibrium]
	Given a strategic form game $\Gamma = \langle \mathcal{A}, (\Pi_i): S \rightarrow S, (\hat{u}_i)\rangle$, the strategy profile $\pi^* = (\pi^*_1, \pi^*_2,..., \pi^*_n)$ is called a PSNE of $\Gamma$ if
	$$\pi^*_i = argmax_{\pi_i \in \Pi_i} \hat{u}_i(\pi_i, \pi^*_{-i}) ~ \forall i \in N$$
	That is, each player's Nash equilibrium strategy is a best response to the Nash equilibrium strategies of the other players~\cite{narahari2014game}.
\end{definition}

\begin{definition}[Nash Incentive Compatibility]
	A social choice function $u: S_1 \times ... \times S_n$ is said to be Nash Incentive Compatible (or truthfully implementable in Nash equilibrium) if the direct revelation mechanism $\mathcal{D} = ((S_i)_{i \in \mathcal{A}}, \hat{u}(.))$ has a pure strategy Nash equilibrium $\pi^{*}(.) = (\pi^{*}_1(.), ..., \pi^{*}_n(.))$ in which $\pi^{*}_i(s_i) = s_i \forall s_i \in S_i, \forall i \in \mathcal{A}$.

	That is, directly asking the agents to report their types and using this information in $\hat{u}(.)$ to get the social outcome will solve both the problems, namely, preference elicitation and preference aggregation \cite{narahari2014game} \cite{garg2008foundations} \cite{garg2008foundations-Pt2}.
\end{definition}

\subsection{Report Strength} \label{sec:Phi}
We define the strength of a report as the count of agents who have reported the same signal as agent $i$. If there is a wide spectrum of possible signals, we discretize the range in buckets and consider two signals as equal, if they belong in the same bucket. Report strength is the first of the three scores we compute for a report. It only depends on the current performance of the agent and hence can be considered as the immediate reward for an agent in a particular query.

\subsubsection{Computation of Report Strength} 
Every $s \in S$ represents a report type and each report type receives certain number of reports. Let $\Phi$ be a function that counts the number of reports of the same type. $\Phi_{r_i}$ is defined as the strength of agent $i$'s report by measuring how many agents are in agreement with agent $i$.
\begin{eqnarray}
	\Phi(r_i) = \sum_{j \in \mathcal{A}}1_{r_i = r_j}
\end{eqnarray}
$\Phi(\cdot)$ is a function on a generic report type and is independent of the agent.
Thus, for computational efficiency, strength score should be pre-computed for every signal which can be done by maintaining a counter for every signal while iterating over all agent reports. That is for $n$ agents, total computational complexity is $O(n)$. Individual strength score for any agent can then be obtained by just referring to the strength score of the corresponding signal reported. 

\begin{observation} \label{observation:Phi-positive}
	The report strength of the report of an agent is always positive.
\end{observation}
\begin{proof}
	For any agent $i$, $\Phi(r_i)$ compares $r_i$ with all the reports including $r_i$ itself, hence ensuring that the report strength is at least 1.
	$$ i \in \mathcal{A} \Rightarrow \Phi(r_i) \geq 1 $$
\end{proof}

\begin{lemma} \label{lemma:Phi}
	Consider a game induced by $u^\Phi$ where $u^\Phi = \Phi$ is the sub-utility function corresponding to the report strength ($\Phi$); and let $\hat{s}$ be the observed signal by all agents and $r_i$ be the report submitted by any agent $i$.
	Then under Nash equilibrium,
	$$
		u^{\Phi}_i(r_i = \hat{s}, r_{-i} = \hat{s}) > u^\Phi_i(r_i \ne \hat{s}, r_{-i} = \hat{s}) \forall r_i \in S, \forall i \in \mathcal{A}
	$$
	That is, if all other agents were to report truthfully, the best response for agent $i$ in order to maximize her sub-utility $u^\Phi_i$ is also to report truthfully.
\end{lemma}

\begin{proof}
	Let $\hat{s}$ be the observed signal by all agents and let $r_i$ be the reported signal of agent $i$. As per Nash equilibrium, assuming every other agent reports honestly i.e. $r_{-i} = \hat{s}$ , the strength of agent $i$'s report is given by,
	\begin{eqnarray*}
		\Phi(r_i) = 1 + (\abs{\mathcal{A}}-1)\times 1_{r_i = s}
	\end{eqnarray*}
	Consequently, sub-utility function $u^{\Phi}$ can be written as,
	\begin{eqnarray*}
		u^{\Phi}_i(r_i \neq \hat{s}, r_{-i} = \hat{s}) &=& 1 \\
		u^{\Phi}_i(r_i = \hat{s}, r_{-i} = \hat{s}) &=& \abs{\mathcal{A}}
	\end{eqnarray*}
	Hence, truthful reporting i.e. $r_i = \hat{s}$ is the best response for agent to maximize $u^\Phi_i$ when all other agents are also reporting truthfully.
\end{proof}

\subsubsection{Selective Fairness}
Information aggregation involves evaluating reported signals from agents. It is difficult to evaluate an agent since there is no ground truth available in most scenarios. Existing peer prediction based mechanisms discussed in \Cref{sec:related-works} rate an agent's report against the report of a random peer agent $p_i$.
This, however, exposes agent $i$ to unfair evaluation because if the peer agent $p_i$ is dishonest or misbehaving and submits false report, it will reflect in worse reward for agent $i$ even if she submits truthful report.
\cite{dwork2012fairness} state that an algorithm is fair if it generates similar results for agents with similar attributes. This is not possible with peer prediction based mechanisms.

Existing mechanisms in the literature mainly discuss settings which follow an \emph{independent-value model} where agents observe varying signals for the same query. This makes it difficult to rate all agents against equivalent reports.
In our case, we follow a \emph{common-value model} where all agents observe the same signal. This makes it easier to rate all agents against the same outcome.

\begin{definition}[Selective Fairness]
	Let $\mathcal{A}$ be a set of agents in the vicinity of subject $QS$ and agents $i, j  \in \mathcal{A}$ be any two arbitrary agents who submit two identical reports $r_i$ and $r_j$ such that $r_i = r_j$. The utility function of a payment scheme $\tilde{u}$, admits \emph{selective fairness} if,
	$$ \tilde{u}_i(r_i, r_{-i}) = \tilde{u}_{j}(r_j, r_{-j}), \forall i,j \in \mathcal{A} $$
\end{definition}

\begin{claim} \label{claim:Phi-fair}
	Report strength is a selectively fair computation score for agents.
\end{claim}

\begin{proof}
	Strength $\Phi(.)$ is a function of a generic report type (i.e. signal) and for a generic signal $s$, it is given as,
	\begin{eqnarray}
		\Phi(s) = \sum_{j \in \mathcal{A}} 1_{s = r_j}
	\end{eqnarray}
	Hence, by construction, agents with similar report will share similar strength score and so $\Phi$ is selectively fair.
\end{proof}

\subsection{Consistency Score (\texorpdfstring{$\alpha$}{alpha})} \label{sec:alpha}
The report of an agent with a history of reporting truthfully is more valuable to FaRM than the report of an agent with a lousy history of reporting. Therefore, instead of distributing rewards based only on the performance of agents in the last query, we keep a metric of every agent's performance so far in the mechanism and take it into consideration while paying the agents. We call this metric \emph{consistency score}.


Consistency score depends on the strength of agent $i$'s report. After every report of agent $i$, the FaRM updates agent's consistency score. The score increases for accurate reporting and gets penalized for inaccurate reporting. It should be noted, however, that the accuracy is relative since the ground truth is not known and the highest reported signal is used as a proxy for ground truth.
Consistent agents receive higher rewards than less consistent ones for the same report submitted to the system because of the higher value of their report. We, therefore, motivate agents to maintain the quality of their reports in order to maintain high rewards. This also encourages agents who have a higher rate of submitting truthful reports to the system to continue participating and at the same time discourages free-riders and dishonest reporters.

\subsubsection{Computation of Consistency Score} 
The consistency score is incremented with respect to the second highest reported signal (denoted by $\phi_2$) and is decremented based on the highest reported signal (denoted by $\phi_1$).
\begin{eqnarray}
	\varphi_1 &=& \max_{s \in S}(\Phi(s)) \\
	\varphi_2 &=& \left\{ \begin{array}{ccl}
		\max2_{s \in S}(\Phi(s))         & \mbox{if} & \max2_{s \in S}(\Phi(s)) > 0 \\\\
		\frac{\varphi_1^2 - 1}{\varphi_1} & \mbox{if} & \max2_{s \in S}(\Phi(s)) = 0
	\end{array}
	\right.
\end{eqnarray}
where $\max2$ gives the second highest value and $\varphi_1 \geq \varphi_2 > 0$.
$\varphi_1$ and $\varphi_2$ can be pre-computed before computing individual agent scores. For $n$ agents, the complexity is $O(n)$. Only agents with the strongest reports will have their consistency score incremented. However, it is not enough to submit one of the strongest reports, the report type must hold percentage strictly greater than other report types. Other agents will have their consistency score decremented. It is updated as follows,
\begin{equation}
	\alpha^t_i = \left\{
	\begin{array}{ccl}
		\alpha^{t-1}_i - \frac{\alpha^{t-1}_i}{k}\times\frac{(\varphi_1 - \Phi(r_i))}{\abs{\mathcal{A}}}   & \mbox{if} & \Phi(r_i) < \varphi_1 \\\\
		\alpha^{t-1}_i + \frac{1-\alpha^{t-1}_i}{k}\times\frac{(\varphi_1 - \varphi_2)}{\abs{\mathcal{A}}} & \mbox{if} & \Phi(r_i) = \varphi_1,
	\end{array}
	\right.
\end{equation}
where $k \geq 1$. Note that $k$ is just an arbitrary constant which can be used to tweak the reward rate. $t$ represents the sequence of rounds agent $i$ has participated. Once $\varphi_1$ and $\varphi_2$ are pre-computed, updating $\alpha$ has constant computational complexity per agent (i.e., the overall complexity is $O(n)$).

\begin{claim} \label{claim:alpha-bounded}
	Consistency score is bounded in the range [0,1).
\end{claim}
\begin{proof}
	Consistency score for any agent is initialized at 0. The score remains 0 until the agent report coincides with the highest reported signal for the first time.
	Let
	\begin{align}
		\alpha^{t-1}_i \in [0,1) \label{eq:alpha-bounded-1}
	\end{align}
	\begin{description}
		\item[Case 1: $\Phi(r_i) < \varphi_1$]
			\begin{align}
				\alpha^t_i &= \alpha^{t-1}_i (1 - \frac{\varphi_1 - \Phi(r_i)}{\abs{\mathcal{A}}}) \label{eq:alpha-bounded-2}
			\end{align}
			Also, we know the following,
			\begin{align}
				\Phi(r_i) < \varphi_1 &< 	\abs{\mathcal{A}} \label{eq:alpha-bounded-3} \\
				k 					  &\geq 1 \label{eq:alpha-bounded-4} \\
			\end{align}
			Using \Cref{eq:alpha-bounded-1,eq:alpha-bounded-3,eq:alpha-bounded-4} in \Cref{eq:alpha-bounded-2}, we get,
			\begin{align}
				\alpha^t_i &\geq  0 \label{eq:alpha-bounded-final1} \\
				\alpha^t_i < \alpha^{t-1}_i &< 1 \label{eq:alpha-bounded-final2}
			\end{align}
		\item[Case 2: $\Phi(r_i) = \varphi_1$]
			\begin{align}
				\alpha^t_i = \alpha^{t-1}_i + \frac{1-\alpha^{t-1}_i}{k}\times\frac{(\varphi_1 - \varphi_2)}{\abs{\mathcal{A}}} \label{eq:alpha-bounded-5}
			\end{align}
			Let $a = \frac{\varphi_1 - \varphi_2}{k\abs{\mathcal{A}}}$ and we already know the following,
			\begin{align}
				0 < \varphi_2 \leq \varphi_1 &\leq \abs{\mathcal{A}} \label{eq:alpha-bounded-6} \\
				k 					  &\geq 1 \label{eq:alpha-bounded-7}
			\end{align}
			From \Cref{eq:alpha-bounded-6,eq:alpha-bounded-7}, we can say
			\begin{align}
				a \in [0,1) \label{eq:alpha-bounded-8}
			\end{align}
			We now rewrite \Cref{eq:alpha-bounded-5} as,
			\begin{align}
				\alpha^t_i = \alpha^{t-1}_i + (1 - \alpha^{t-1}_i)\cdot a \label{eq:alpha-bounded-9}
			\end{align}
			Using \Cref{eq:alpha-bounded-8,eq:alpha-bounded-1} in \Cref{eq:alpha-bounded-8}, we can say,
			\begin{align}
				\alpha^t_i \geq \alpha^{t-1}_i &\geq 0 \label{eq:alpha-bounded-final3} \\
				\alpha^t_i 					&< 1 \label{eq:alpha-bounded-final4}
			\end{align}
			From \Cref{eq:alpha-bounded-final1,eq:alpha-bounded-final2,eq:alpha-bounded-final3,eq:alpha-bounded-final4}, we prove our claim that $\alpha(\cdot) \in [0,1)$
				
	\end{description}
\end{proof}

Notice that the consistency score increases rapidly when it is closer to zero and it decreases fast when it is closer to one. This property motivates new agents to join the system and prevent existing agents from downgrading the quality of their reports.

\begin{lemma} \label{lemma:alpha}
	Consider a game induced by $u^\alpha$ where $u^\alpha = \alpha$ is the sub-utility function corresponding to consistency score ($\alpha$); and let $\hat{s}$ be the observed signal by all agents and $r_i$ be the report submitted by any agent $i$.
	Then considering the property of Nash equilibrium,
	$$
		u^{\alpha}_i(r_i = \hat{s}, r_{-i} = \hat{s}) > u^\alpha_i(r_i \ne \hat{s}, r_{-i} = \hat{s}) \forall r_i \in S, \forall i \in \mathcal{A}
	$$
	That is, if all other agents were to report truthfully, the best response for agent $i$ in order to maximize her sub-utility $u^\alpha_i$ is also to report truthfully.
\end{lemma}

\begin{proof}
	Let $\hat{s}$ be the observed signal by all agents and let $r_i$ be the reported signal by agent $i$. We assume every other agent to report truthfully as per Nash equilibrium, i.e. $r_{-i} = \hat{s}$.
	Then, $\varphi_1$ and $\varphi_2$ are given as,
	\begin{eqnarray*}
		\varphi_1 &=& \abs{\mathcal{A}} - 1_{r_i \ne s} \\
		\varphi_2 &=& \left\{ \begin{array}{ccl}
			1                                                 & \mbox{if} & r_i \ne s \\
			\frac{\abs{\mathcal{A}}^2 - 1}{\abs{\mathcal{A}}} & \mbox{if} & r_i = s
		\end{array}
		\right.
	\end{eqnarray*}
	Using above simplifications, we derive simplified consistency score for agent $i$ as,
	$$
		\alpha^t_i = \left\{
		\begin{array}{ccl}
			\alpha^{t-1}_i - \frac{\alpha^{t-1}_i}{k}\times\frac{(\abs{\mathcal{A}}-2)}{\abs{\mathcal{A}}}   & \mbox{if} & r_i \ne s \\\\
			\alpha^{t-1}_i + \frac{1-\alpha^{t-1}_i}{k}\times\frac{(1/\abs{\mathcal{A}})}{\abs{\mathcal{A}}} & \mbox{if} & r_i = s
		\end{array}
		\right.
	$$
	Consequently, the sub-utility function $u^\alpha_i$ can be written as,
	\begin{eqnarray*}
		u^{\alpha}_i(r_i \neq \hat{s}, r_{-i} = \hat{s}) &=& \alpha^{t-1}_i - \frac{\alpha^{t-1}_i}{k}\times\frac{(\abs{\mathcal{A}}-2)}{\abs{\mathcal{A}}} \\
		u^{\alpha}_i(r_i = \hat{s}, r_{-i} = \hat{s}) &=& \alpha^{t-1}_i + \frac{1-\alpha^{t-1}_i}{k}\times\frac{(1/\abs{\mathcal{A}})}{\abs{\mathcal{A}}}
	\end{eqnarray*}
	Given $\abs{\mathcal{A}} \geq 3$ and $\alpha^{t-1} \in [0,1)$ (from Claim~\ref{claim:alpha-bounded}),
	\begin{equation*}
		u^{\alpha}_i(r_i = \hat{s}, r_{-i} = \hat{s}) > u^{\alpha}_i(r_i \neq \hat{s}, r_{-i} = \hat{s})
	\end{equation*}
	Hence, truthful reporting i.e. $r_i = \hat{s}$ is the best response for agent to maximize $u^\alpha_i$ when all other agents are also reporting truthfully.
\end{proof}

\subsubsection{Cumulative Fairness}
Report from an agent who is consistently reporting truthful information is more valuable than the report of an agent with a lousy history of reporting.

\textbf{Example:} In her summer vacation, Anna joins professional cooking classes out of interest. She likes to frequently practice recipes at home and get reviews from her family and neighbors. Once she makes Turkish delights and has Barbara and Chris independently assess the delicacy. However, they submit contrasting reviews. But Anna recalls Barbara's previous reports and finds her to be misreporting on various occasions in order to make her happy, while Chris has been almost always critical but truthful in his reports. Thus, based on the history of reports of Barbara and Chris, Anna identifies Chris' reports to be more useful than that from Barbara.

It's only fair to take the value of an agent's report into consideration while distributing reward. Motivated by similar analogy, we incorporated consistency score ($\alpha$) in our final reward structure.

\begin{definition}[Cumulative Fairness] \label{def:cumulative-fairness}
	Let $\mathcal{A}$ be a set of agents in the vicinity of local subject $QS$ and agents $i, j \in \mathcal{A}$ be any two arbitrary agents submitting report $r_i$ and $r_j$ such that $r_i = r_j$. Then the utility function $\hat{u}(.)$ of an arbitrary payment scheme is said to admit cumulative fairness if,
	\begin{eqnarray*}
		\hat{u}^t_i(r_i = \hat{s}, r_{-i}, \hat{u}^{t-1}_i) > \hat{u}^t_j(r_j = \hat{s}, r_{-j}, \hat{u}^{t-1}_j) \\ \forall i,j \in \mathcal{A} ~ s.t. ~ \hat{u}^{t-1}_i > \hat{u}^{t-1}_j
	\end{eqnarray*}
\end{definition}

\begin{claim} \label{claim:alpha-fair}
	Consistency score is a cumulatively fair computation score for agents.
\end{claim}

\begin{proof}
	Let $\mathcal{A}$ be a set of agents in the vicinity of local subject $QS$ and agents $i, j \in \mathcal{A}$ be any two arbitrary agents submitting report $r_i$ and $r_j$ such that $r_i = r_j$ and $\alpha^{t-1}_i > \alpha^{t-1}_j$.
	
	Since there reports are same and strength ($\Phi(.)$) is a selectively fair computation score (from Claim~\ref{claim:Phi-fair}), we can say, $\Phi(r_i) = \Phi(r_j)$.

	\begin{description}
		\item[Case 1: $\Phi(r_i = r_j) < \varphi_1$]
			\begin{align*}
				\alpha^t_i &= \alpha^{t-1}_i - \frac{\alpha^{t-1}_i}{k}\times\frac{(\varphi_1 - \Phi(r_i))}{\abs{\mathcal{A}}} \\
				\alpha^t_j &= \alpha^{t-1}_j - \frac{\alpha^{t-1}_j}{k}\times\frac{(\varphi_1 - \Phi(r_j))}{\abs{\mathcal{A}}}
			\end{align*}
			Let $a_1 = \frac{(\varphi_1 - \Phi(r_i))}{\abs{\mathcal{A}}} = \frac{(\varphi_1 - \Phi(r_j))}{\abs{\mathcal{A}}}$. \\\\
			Then, we can rewrite $\alpha_i, \alpha_j$ as,
			\begin{align}
				\alpha^t_i &= \alpha^{t-1}_i (1 - \frac{a_1}{k}) \label{eq:alpha-fair-1} \\
				\alpha^t_j &= \alpha^{t-1}_j (1 - \frac{a_1}{k}) \label{eq:alpha-fair-2}
			\end{align}
			Also, we know the following:
			\begin{align}
				\Phi(r_i = r_j) < \varphi_1 < \abs{\mathcal{A}} \Rightarrow a_1 &< 1 \label{eq:alpha-fair-3} \\
				k &\geq 1 \label{eq:alpha-fair-4} \\
				\alpha^{t-1}_i &> \alpha^{t-1}_j \label{eq:alpha-fair-5}
			\end{align}
			Using \Cref{eq:alpha-fair-3,eq:alpha-fair-4,eq:alpha-fair-5} in \Cref{eq:alpha-fair-1,eq:alpha-fair-2}, we get,
			\begin{align}
				\alpha^t_i > \alpha^t_j \label{eq:alpha-fair-final1}
			\end{align}
		\item[Case 2: $\Phi(r_i = r_j) = \varphi_1$]
			\begin{align*}
				\alpha^t_i &= \alpha^{t-1}_i + \frac{1-\alpha^{t-1}_i}{k}\times\frac{(\varphi_1 - \varphi_2)}{\abs{\mathcal{A}}} \\
				\alpha^t_j &= \alpha^{t-1}_j + \frac{1-\alpha^{t-1}_j}{k}\times\frac{(\varphi_1 - \varphi_2)}{\abs{\mathcal{A}}}
			\end{align*}
			Let $a_2 = \frac{(\varphi_1 - \varphi_2)}{\abs{\mathcal{A}}}$. \\\\
			Then, we can rewrite $\alpha_i, \alpha_j$ as,
			\begin{align}
				\alpha^t_i &= \alpha^{t-1}_i (1 - \frac{a_2}{k}) + \frac{a_2}{k} \label{eq:alpha-fair-6} \\
				\alpha^t_j &= \alpha^{t-1}_j (1 - \frac{a_2}{k}) + \frac{a_2}{k} \label{eq:alpha-fair-7}
			\end{align}
			Also, we know the following:
			\begin{align}
				0 < \varphi_2 \leq \varphi_1 \leq \abs{\mathcal{A}} \Rightarrow a_2 &< 1 \label{eq:alpha-fair-8} \\
				k &\geq 1 \label{eq:alpha-fair-9} \\
				\alpha^{t-1}_i &> \alpha^{t-1}_j \label{eq:alpha-fair-10}
			\end{align}
			Using \Cref{eq:alpha-fair-8,eq:alpha-fair-9,eq:alpha-fair-10} in \Cref{eq:alpha-fair-6,eq:alpha-fair-7}, we get,
			\begin{align}
				\alpha^t_i > \alpha^t_j \label{eq:alpha-fair-final2}
			\end{align}
	\end{description}
	From \Cref{eq:alpha-fair-final1,eq:alpha-fair-final2} and Definition~4.9, we say that consistency score admits cumulative fairness.
\end{proof}

\subsection{Reliability Score (\texorpdfstring{$\beta$}{beta})} \label{sec:beta}
Reliability score provides incentives to agent $i$ to not collude with her nearby agents $\mathcal{I}_i$. It is defined as the ratio of \emph{external agreement} by \emph{internal agreement}. The internal agreement is the percentage of nearby agents (formally identified as \emph{internal peers} $\mathcal{I}_i$) that share the same report as agent $i$ and the external agreement is the percentage of agents other than nearby agents (formally addressed as \emph{external peers} $\mathcal{E}_i$) who share the same report as agent $i$.
\begin{definition}[Internal and External Peers]
	The \emph{internal} peers of agent $i$ (denoted by $\mathcal{I}_i$) are the ones with which $i$ can exchange a message using a device-to-device communication protocol like Wi-Fi direct or Bluetooth. The internal peers combined with the agent's location compose her context $l_i$. Every agent $j$ that is not directly accessible from $i$ (i.e., $j \notin \mathcal{I}_i$) belongs to her \emph{external} peers denoted by $\mathcal{E}_i$.
\end{definition}

\subsubsection{Computation of Reliability Score} 

For agent $i$, her reliability score is computed as:
\begin{equation}
 	\beta_i = \frac{\frac{(\sum_{j \in \mathcal{E}_i}1_{r_i=r_j})}{\abs{\mathcal{E}_i}}}{\frac{(\sum_{j \in \mathcal{I}_i}1_{r_i=r_j})}{\abs{\mathcal{I}_i}}+1}
\end{equation}
For $n$ agents, the computational complexity for calculating $\beta$ is $O(n)$ per agent.

\begin{claim} \label{claim:beta-bounded}
	Reliability score is bounded in the range [0,1)
\end{claim}
\begin{proof}
	Let
	\begin{align}
		b_1 = (\sum_{j \in \mathcal{E}_i}1_{r_i=r_j})~/~\abs{\mathcal{E}_i} \\
		b_2 = \left[(\sum_{j \in \mathcal{I}_i}1_{r_i=r_j})~/~\abs{\mathcal{I}_i}\right]+1
	\end{align}
	By construction, $b_1 \in [0,1]$ and $b_2 \in [1,2]$

	Hence, $\beta_i = \frac{b_1}{b_2} \in [0,1]$ 
\end{proof}

\begin{lemma} \label{lemma:beta}
	Consider a game induced by $u^\beta$ where $u^\beta = \beta$ is the sub-utility function corresponding to reliability score ($\beta$); and let $\hat{s}$ be the observed signal by all agents and $r_i$ be the report submitted by any agent $i$.
	Then considering the property of Nash equilibrium,
	$$
		u^{\beta}_i(r_i = \hat{s}, r_{-i} = \hat{s}) > u^\beta_i(r_i \ne \hat{s}, r_{-i} = \hat{s}) \forall r_i \in S, \forall i \in \mathcal{A}
	$$
	That is, if all other agents were to report truthfully, the best response for agent $i$, in order to maximize her sub-utility $u^\beta_i$, is also to report truthfully.
\end{lemma}

\begin{proof}
	Let $\hat{s}$ be the signal observed by all agents and $r_i$ be the reported signal of agent $i$. We assume every other agent to report truthfully as per Nash equilibrium. Then, the reliability score for agent $i$'s report can be simplified as,
	\begin{eqnarray*}
		\beta_i &=& \frac{1_{r_i = \hat{s}}}{1_{r_i = \hat{s}} + 1} \\
		&=& 0.5\times 1_{r_i = \hat{s}}
	\end{eqnarray*}
	Hence, agent $i$ maximizes her reliability score by reporting truthfully, i.e. $r_i = \hat{s}$.
\end{proof}


\begin{claim}
	Reliability score prevents agents from colluding with nearby agents.
\end{claim}
\begin{proof}[Proof Sketch]
	Reliability score ($\beta$) can be seen as,
	$$\beta = \frac{\mbox{external agreement}}{\mbox{internal agreement} + 1} $$
	where external agreement is the percentage of the external peers that agree with the agent and internal agreement is the percentage of the internal peers that agree with the agent. For the sake of convenience we use variables $l_i$ and $g_i$ as,
	\begin{eqnarray*}
		l_i &=& \mbox{internal agreement} 
		=  \frac{1}{\abs{\mathcal{I}_i}} \sum_{j \in \mathcal{I}_i}1_{r_i=r_j}\\
		g_i &=& \mbox{external agreement} 
		= \frac{1}{\abs{\mathcal{E}_i}} \sum_{j \in \mathcal{E}_i}1_{r_i=r_j}
	\end{eqnarray*}
	Let's consider a malicious colluding agent $m$. We assume agents around $QS$ are reasonably scattered and that the majority of agents are honest. In that case, agent $m$ will have high internal agreement and less external agreement i.e. $l_m > g_m$. Since, $\beta_m \propto g_m$ and $\beta_m \propto 1/(l_m+1)$, the more agent $m$ colludes with internal peers the lower her reliability score drops. In case where nearby agents are also honest, it is still in favor of agent $m$ to submit truthful report since she will want to maintain her reporting agreement with external peers.
	
	On the other hand, if agent $m$ were to be honest while her internal peers collude, the outcome is $l_m < g_m$. This leads to increase in reliability score. Hence, if internal peers of agent $m$ are colluding, it is profitable for agent $m$ to remain honest. Thus, reliability score prevents agents from getting swayed by nearby agents (internal peers).
\end{proof}

\subsection{Location Robustness Score (\texorpdfstring{$\gamma$}{gamma})} \label{sec:gamma}

Apart from the aforementioned scores, we introduce the location robustness score of agent $i$, $\gamma_i$, to detect whether the location of a mobile agent is close to the location of the query subject. Every mobile agent can exchange messages with her internal peers in order to justify that she is within a distance from the query subject~\cite{Chatzopoulos:2016:LAP:2942358.2947401}. In detail, considering that the context of agent $i$, $l_i$ is composed of her location $l_{i}^{Loc}$ and her nearby agents $\mathcal{I}_i$, agent $i$ can send a message to every node $j \in \mathcal{I}_i$ with $l_{i}^{Loc}$ in it using a communication technology like Bluetooth or WiFi-direct. Every agent $j \in \mathcal{I}_i$ will calculate the distance with agent $i$, $D_{ji}$ and will respond to the message by including her location $l_{j}^{Loc}$ and the calculated distance. Agent $j$ can sign her message with a private key in order to make sure that agent $i$ will not alternate her estimation of the difference between the two locations. Based on these messages, agent $i$ will calculate the robustness of her location. 
\begin{equation}
	\gamma_i = \frac{1}{\abs{\mathcal{I}_i}}\sum_{j \in \mathcal{I}_i}1_{D_{ji} \leq D^{T}}
\end{equation}
Which is practically the fraction of their neighbors who are within a distance $D^{T}$ that depends on the communication technology. For example, two mobile agents can communicate with Bluetooth if $D^{T}<20$ meters and with WiFi-direct if $D^{T}<50$ meters. Depending on the query subject and the predetermined allowed distance from which a mobile agent can produce a report, $D^{T}$ may be smaller than the values achieved by the employed communication technology. FaRM uses location robustness to detect fraudulent agents.

Every participating mobile agent, has to submit $\gamma_i$ together with her report. Given that a fraudulent agent is not able to alternate $\gamma_i$, since every message from her internal peers is signed, FaRM filters reports with low $\gamma$ values.

\subsection{Rewards}
Once FaRM filters reports using the location robustness score, it updates consistency and reliability scores for each agent, then finally it calculates every agent's reward for that query.
The location robustness score ($\gamma$) is only used for filtration and not considered in the final reward. The final reward of agent $i$ depends on the product of $\Phi(r_i), \alpha_i$ and $\beta_i$.
The strength of agent's report factors in the immediate reward for the respective query. The consistency score takes into account agent's tendency to report truthfully and therefore motivate a more consistent agent to participate more. The reliability score keeps the agent from colluding with other agents. An agent who has a history of colluding with nearby agents will have less motivation to participate in the system, hence preventing noise generated by misleading reports. The final utility score ($u_i$) is computed as:
\begin{equation}
	u_i(r_i, r_{-i}) = \frac{u^\Phi_i(r_i, r_{-i})~.~u^\alpha_i(r_i, r_{-i})~.~u^\beta_i(r_i, r_{-i})}{\abs{\mathcal{A}}^2} \times B,
\end{equation}
where $B$ is the reward budget per query.
Also, since final utility is the product of all three scores, even if some closely located agents collude in order to increase report strength, they risk decreasing reliability score which can mask the benefits of collusion and result in even worse payoffs.

\begin{theorem} \label{theorem:NIC}
	FaRM is Nash incentive compatible with guaranteed non-negative payoffs and weak budget balanced.
\end{theorem}

\begin{proof}
	Let $\hat{s}$ be the signal observed by all agents and $r_i$ be the reported signal of agent $i$. We assume all other agents to report truthfully as per Nash equilibrium i.e. $r_{-i} = \hat{s}$.\\
	Let $$ w_i(r_i, r_{-i}) = u^\Phi_i(r_i, r_{-i})~.~u^\alpha_i(r_i, r_{-i})~.~u^\beta_i(r_i, r_{-i}) $$ and $$ W_{-i}(r_i, r_{-i}) = \sum_{j \in \mathcal{A}\backslash i} w_j(r_j, r_{-j}), $$ then we can rewrite $u_i$ as:
	\begin{align*}
		u_i &= \frac{w_i}{\abs{\mathcal{A}}^2}\times B
	\end{align*}
	For a fixed budget $B$ and $\abs{\mathcal{A}}$ number of agents participating, $$ argmax_{r_i}(u_i) \Rightarrow argmax_{r_i}(w_i) $$ 

	Also, $w_i$ can be written for Nash equilibrium case as,
	\begin{align*}
		w_i(r_i \neq \hat{s}, r_{-i} = \hat{s}) = & u^\Phi_i(r_i \neq \hat{s}, r_{-i} = \hat{s})~.\\
		& u^\alpha_i(r_i \neq \hat{s}, r_{-i} = \hat{s})~.~u^\beta_i(r_i \neq \hat{s}, r_{-i} = \hat{s}) \\
		w_i(r_i = \hat{s}, r_{-i} = \hat{s}) = & u^\Phi_i(r_i = \hat{s}, r_{-i} = \hat{s})~.\\
		& u^\alpha_i(r_i = \hat{s}, r_{-i} = \hat{s})~.~u^\beta_i(r_i = \hat{s}, r_{-i} = \hat{s})
	\end{align*}

	Using Lemmas~\ref{lemma:Phi}, \ref{lemma:alpha} and \ref{lemma:beta}, we can say $w_i(r_i = \hat{s}, r_{-i} = \hat{s}) > w_i(r_i \neq \hat{s}, r_{-i} \hat{s})$\\
	
	Hence, agent $i$ gains maximum utility when she reports truthfully i.e. $r_i = s$.

	Also since, $\Phi(r_i) > 0$ (from Observation~\ref{observation:Phi-positive}), $\alpha_i \in [0,1)$ (Claim~\ref{claim:alpha-bounded}), $\beta_i \in [0,1]$ (Claim~\ref{claim:beta-bounded}), we can say $(w_i \geq 0 \Rightarrow u_i \geq 0) ~ \forall \mathcal{A}$, irrespective of them reporting truthfully. For truthful reporting in Nash equilibrium, FaRM guarantees strictly positive payoffs which can be easily verified.
	Hence, FaRM guarantees non-negative utilities for all agents, however, an agent must report truthfully to maximize her utility and maintain her consistency score. This motivates agents to report truthfully but at the same time protects them from loss in case of reporting error.

	Furthermore,
	\begin{align*}
		\Phi(r_i) \in [1,\abs{\mathcal{A}}] \Rightarrow &~ u^\Phi_i \in [1,\abs{\mathcal{A}}] \\
		\alpha_i \in [0,1) \Rightarrow &~ u^\alpha_i \in [0,1) \\
		\beta_i \in [0,1] \Rightarrow &~ u^\beta_i \in [0,1]
	\end{align*}
	From above three equations, $max_{r_i}(w_i) < \abs{\mathcal{A}}$. Consequently, $$ max_{r_i}(u_i) < \frac{B}{\abs{\mathcal{A}}} $$
	And hence, in any case scenario it is not possible for the total reward of $\abs{\mathcal{A}}$ agents to exceed budget $B$. 
	Thus, FaRM is weak budget balanced.

\end{proof}
 
\begin{proposition}
	FaRM admits selective fairness and cumulative fairness and hence is a fair mechanism. 	 
 \end{proposition}
\begin{proof}
	From \Cref{claim:Phi-fair}, we show FaRM admits selective fairness for immediate query.
	Also, from \Cref{claim:alpha-fair}, we show that FaRM ensures agents are accounted for their consistency and valuable reports.
	Hence, we conclude FaRM to be a fair mechanism.
\end{proof}

\section{Conclusion}
Existing information aggregation mechanisms are focused on settings relevant to mass population, e.g., online reviews about products and services, and community sensing regarding prevalent societal topics (pollution, global warming, etc.). FaRM focuses on location-specific queries of people in their everyday life. It can be challenging to obtain prior distribution models for the signal space and hence, prior-free mechanisms are preferable for these set of queries. 
In spontaneous localized settings all the agents are assumed to observe the same signal. FaRM leverages this property to provide fair rewards to all agents.
FaRM's payment scheme consists of three sub-utility structures, namely \emph{report strength}, \emph{consistency score} and \emph{reliability score}.
Report strength ensures \emph{selective fairness} for the immediate query by evaluating all agents equivalently against the same result.
Consistency score ensures \emph{cumulative fairness} by accounting for an agent's consistency in the system. The report of an agent consistently reporting truthful reports to the system is more valuable than the report of an agent with a lousy reporting history. Consistency score takes the value of a report into consideration. Therefore, it helps in preventing agents from free-riding as it can affect their consistency score negatively.
Reliability score estimates an agent's collusion with the nearby agents and reflects the same in the final reward. It motivates agents to earn a higher reward by not colluding with nearby agents and instead stick with their observed signal.
Together these scores ensure fair reward to the agents, not just for their immediate contribution but also for their consistent and reliable service to the system.
In summary, FaRM is a Nash incentive compatible mechanism which rewards agents fairly for their services. 

\clearpage

\bibliographystyle{named}
\bibliography{ref}

\end{document}